\documentclass{article}
\usepackage{latexsym,amssymb,enumerate,amsmath,amsthm} 
  \usepackage{graphicx} 
  \usepackage{tikz}
  \usepackage[square,sort,comma,numbers]{natbib}
  \usepackage{hyperref}
\usepackage{latexsym,amssymb,enumerate,amsmath} 
\usepackage{pgfplots}
\usepackage{xcolor,soul}

\numberwithin{equation}{section}

\def\E{\mathbb{E}}
\def\P{\mathbb{P}}
\def\R{\mathbb{R}}

\def\Z{{\mathbf Z}}

\def\n{{\mathbf n}}
\def\T{\text{T}}
\def\1{{\mathbf 1}}

\def\O{\mathcal{O}}

\newcommand{\dd}{\text{d}}

\def\x{{\mathbf x}}
\def\nf{n_{\textup{f}}}
\def\Z{\widetilde{Z}}

\def\T{\textup{T}}
\def\X{\mathbf{X}}

\def\wtau{\widetilde{\tau}}
\def\oned{\textup{1d}}
\def\twod{\textup{2d}}
\def\threed{\textup{3d}}
\def\side{\textup{side}}
\def\top{\textup{top}}
\def\bot{\textup{bot}}
\def\esc{\textup{esc}}
\def\ball{\textup{ball}}

\newtheorem{theorem}{Theorem}
\newtheorem{lemma}[theorem]{Lemma}

\theoremstyle{plain}
\theoremstyle{remark}

\usepackage{amsfonts}
\usepackage{graphicx}
\usepackage{epstopdf}
\usepackage{algorithmic}
\ifpdf
  \DeclareGraphicsExtensions{.eps,.pdf,.png,.jpg}
\else
  \DeclareGraphicsExtensions{.eps}
\fi

\newcommand{\TheTitle}{A probabilistic approach to extreme statistics of Brownian escape times in dimensions 1, 2, and 3}
\newcommand{\TheAuthors}{Sean D. Lawley and Jacob B. Madrid}

\ifpdf
\hypersetup{
  pdftitle={\TheTitle},
  pdfauthor={\TheAuthors}
}
\fi

\begin{document}

\title{{\TheTitle}
}

\author{Sean D. Lawley\thanks{Department of Mathematics, University of Utah, Salt Lake City, UT 84112 USA (\texttt{lawley@math.utah.edu}).}
\; and Jacob B. Madrid\thanks{Department of Mathematics, University of Utah, Salt Lake City, UT 84112 USA (\texttt{madrid@math.utah.edu}).}
}
\date{\today}
\maketitle

\begin{abstract}
First passage time (FPT) theory is often used to estimate timescales in cellular and molecular biology. While the overwhelming majority of studies have focused on the time it takes a given single Brownian searcher to reach a target, cellular processes are instead often triggered by the arrival of the first molecule out of many molecules. In these scenarios, the more relevant timescale is the FPT of the first Brownian searcher to reach a target from a large group of independent and identical Brownian searchers. Though the searchers are identically distributed, one searcher will reach the target before the others and will thus have the fastest FPT. This fastest FPT depends on extremely rare events and its mean can be orders of magnitude faster than the mean FPT of a given single searcher. In this paper, we use rigorous probabilistic methods to study this fastest FPT. We determine the asymptotic behavior of all the moments of this fastest FPT in the limit of many searchers in a general class of two and three dimensional domains. We establish these results by proving that the fastest searcher takes an almost direct path to the target.
\end{abstract}


\section{Introduction}

Several{ investigations} and commentaries have recently announced a paradigm shift in studying cellular activation rates \cite{basnayake2019, schuss2019, coombs2019, redner2019, sokolov2019, rusakov2019, martyushev2019, tamm2019}. This{ work} has produced new questions, calls for more investigation, and intriguing conjectures to explain the seeming ``redundancy'' that marks many biological systems \cite{schuss2019}. Indeed, this work has led to the formulation of the so-called ``redundancy principle,'' which asserts that many seemingly redundant copies of an object (cells, proteins,{ molecules, }etc.)\ are not a waste, but rather have the specific function of accelerating cellular activation rates \cite{schuss2019}. For example, this principle has been used to explain (i) why thousands of neurotransmitters are released in order to activate only a few receptors in the synaptic cleft \cite{basnayake2019c} and (ii) why $300$ million sperm cells attempt to find the oocyte in human fertilization, when only a single sperm cell is necessary \cite{meerson2015, eisenbach2006, reynaud2015}.

To give the background for this paradigm shift, many cellular processes are triggered when a ``searcher'' reaches a ``target'' \cite{chou_first_2014, holcman_time_2014, bressloff_stochastic_2013}. Some examples are oocyte fertilization by the arrival of a sperm cell \cite{reynaud2015}, calcium release triggering by diffusing IP$_{3}$  molecules that reach IP$_{3}$  receptors \cite{wang1995}, gene activation by the arrival of a diffusing transcription factor to a certain gene \cite{larson2011}, adaptive immune response initiation by T cells binding to antigen presenting cells \cite{edwards2011, delgado2015}, and many more \cite{chou_first_2014}. In these systems, the first passage time (FPT) of a searcher to a target sets the timescale of activation.

This timescale has been estimated by calculating the mean first passage time (MFPT) of a given single searcher to a target. Indeed, many prior studies have calculated such MFPTs, especially for the case of a diffusing Brownian searcher and a small target, which is the so-called narrow escape problem \cite{benichou2008, holcman2014, grebenkov2017, lindsay2017opt, lawley2019dtmfpt, kurella2015, ward10, cheviakov2010, lindsay2017}. However, {the relevant timescale in many systems} is not the MFPT of a given single searcher but rather the MFPT of the fastest searcher out of many searchers \cite{schuss2019, barlow2016, basnayake2019fast, guerrier2018, hartich2019}.{ For important earlier work on such fastest FPTs, see }\cite{weiss1983, yuste1996, yuste1997, yuste_diffusion_2000, yuste2001, van2003, yuste2007, grebenkov2010}.

To illustrate, consider $N$ noninteracting searchers diffusing in a bounded domain $\Omega\subset\R^{d}$ with reflecting boundary $\partial\Omega$. The independent and identically distributed (iid) searchers move by pure Brownian motion with diffusivity $D>0$ and are each initially placed at some $\x_{0}\in\Omega$. 

For these $N$ searchers, let $\tau_{1},\dots,\tau_{N}$ be their $N$ iid FPTs to find a target $\partial\Omega_{\T}\subset\partial\Omega$. Specifically, if the position of the $n$-th searcher at time $t\ge0$ is denoted by $\X_{n}(t)\in\overline{\Omega}$, then
\begin{align*}
\tau_{n}
:=\inf\{t>0:\X_{n}(t)\in\partial\Omega_{\T}\}.
\end{align*}
The MFPT of a given single searcher is $\E[\tau_{n}]=\E[\tau_{1}]$ and its behavior has been studied extensively. However, the more relevant timescale in many biological systems is actually the FPT of the fastest searcher,
\begin{align*}
T
:=\min\{\tau_{1},\dots,\tau_{N}\}.
\end{align*}
That is, $T$ is the first time that any of the $N$ searchers reaches the target.

\begin{figure}
\begin{center}
\includegraphics[width=11cm]{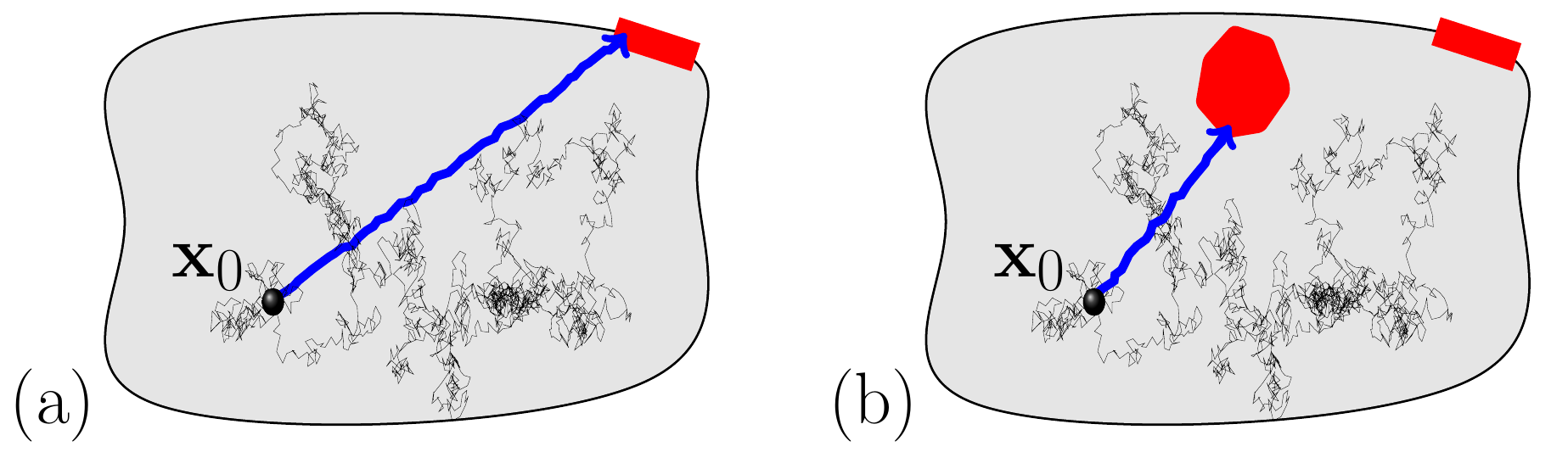}
\caption{
The {thin black} curve illustrates a typical searcher path that wanders around the domain before finding the target. The solid blue curve illustrates that the fastest searcher out of $N\gg1$ searchers takes a straight path to the target. The targets $\partial\Omega_{\T}$ are the red regions and can be on the outer boundary (panel (a) has one outer boundary target) and/or on inner boundaries of the domain (panel (b) has one outer boundary target and one inner boundary target).}\label{figextreme}
\end{center}
\end{figure}

Importantly, if there are many searchers ($N\gg1$) and the target is small, then these two times can be drastically different,
\begin{align*}
\E[T]\ll\E[\tau_{1}].
\end{align*}
The essential reason for this drastic difference in timescales is that $\E[\tau_{1}]$ describes a typical searcher that wanders around the domain before finding the target, while  {$\E[T]$} depends on extremely rare events in which a searcher happens to go directly to the target. Indeed, a significant point that is argued in references \cite{godec2016, godec2016x, basnayake_extreme_2018, schuss2019} is that the path of the fastest searcher closely follows the shortest path from its initial position to the target, see Figure~\ref{figextreme}.

{Compared to the MFPT of a given single searcher, $\E[\tau_{1}]$}, much less is known about the FPT of the fastest searcher, and this is especially true in space dimensions 2 and 3 (2d and 3d). However, it has been known since {1983} \cite{weiss1983} that in one space dimension (1d), the fastest FPT, $T=T_{\oned}$, has the following asymptotic mean behavior as the number of searchers $N$ grows,
\begin{align}\label{1d}
\E[T_{\text{1d}}]
\sim\frac{z_{0}^{2}}{4D\log N}\quad\text{as }N\to\infty,
\end{align}
where $z_{0}$ is the initial distance between the searchers and the target. (Throughout this paper, the notation ``$f\sim g$ as $N\to\infty$'' means $\lim_{N\to\infty}f/g=1$.) More generally, for the spherically symmetric problem of escape from a hypersphere of radius $z_{0}$ in dimension $d\ge1$,
\begin{align*}
\Omega:=\{\x\in\R^{d}:\|\x\|<z_{0}\},\quad
\partial\Omega_{\T}:=\partial\Omega=\{\x\in\R^{d}:\|\x\|=z_{0}\},\quad
\x_{0}=0,
\end{align*}
the 2001 reference~\cite{yuste2001} found the following behavior for the $m$-th moment of the fastest FPT, $T=T_{d\textup{-dim sphere}}$, for any $m\ge1$ and $d\ge1$,
\begin{align}\label{ball}
\E[(T_{d\textup{-dim sphere}})^{m}]
\sim\Big(\frac{z_{0}^{2}}{4D\log N}\Big)^{m}\quad\text{as }N\to\infty.
\end{align}

The interest and development of the ``redundancy principle'' was started by an important recent study that investigated $\E[T]$ in 2d and 3d {bounded} domains \cite{basnayake2019}. In \cite{basnayake2019}, the authors employed formal asymptotic analysis of partial differential equations (PDEs) to argue that the large $N$ behavior of $\E[T]$ in a bounded 2d domain with a small target is the same as the 1d behavior in \eqref{1d}, where $z_{0}>0$ is the distance between the initial searcher position and the target. {The authors employed similar formal arguments to conclude that $\E[T]$ has qualitatively different behavior in bounded 3d domains. Specifically, rather than the $(\log N)^{-1}$ decay seen in 1d and 2d, the authors assert that the mean of the 3d fastest FPT, $T=T_{\threed}$, decays like $(\sqrt{\log N})^{-1}$ as $N\to\infty$.}

In this paper, we use probabilistic methods to prove that the large $N$ behavior of the fastest FPT in a general class of 2d and 3d domains is identical in mean to the 1d behavior in \eqref{1d}. Furthermore, for this general class of 2d and 3d domains, we prove that the asymptotic behavior of the $m$-th moment of the fastest FPT is identical to \eqref{ball}. That is, for any moment $m\ge1$, we prove that
 \begin{align}\label{main}
 \E[(T_{\threed})^{m}]
\sim\E[(T_{\twod})^{m}]
\sim\Big(\frac{z_{0}^{2}}{4D\log N}\Big)^{m}\quad\text{as }N\to\infty.
\end{align}
The general class of 2d and 3d domains requires (i) that the domain contains the straight line path from the initial searcher location to the nearest point on the target and (ii) that a mild so-called star condition holds (see section~\ref{section general} for a precise statement{, but note that a convex domain is a sufficient condition). This corrects the aforementioned result of reference }\cite{basnayake2019},{ which was due to a small error (see their equation~(97)) which followed several pages of innovative formal calculations.}

The rest of the paper is organized as follows. In section~\ref{section cylinder}, we analyze the fastest FPT in a simple cylindrical domain in 3d. In section~\ref{section general}, we extend the methods developed for the cylindrical domain to a general class of 2d and 3d domains. We conclude with a brief discussion {highlighting additional questions about fastest FPTs.}

\section{Simple cylindrical geometry}\label{section cylinder}

Let the spatial domain $\Omega\subset\R^{3}$ be a cylinder of radius $r>0$ and height $h>0$,
\begin{align*}
\Omega:=\big\{(x,y,z)\in\R^{3}:x^{2}+y^{2}<r^{2},\,z\in(0,h)\big\}.
\end{align*}
Suppose that the boundary $\partial\Omega$ contains a distinguished region, $\partial\Omega_{\text{T}}\subset\partial\Omega$, which we refer to as the target, which is a disk of radius $a\in(0,r)$ at the bottom of the cylinder,
\begin{align*}
\partial\Omega_{\text{T}}:=\big\{(x,y,0)\in\R^{3}:x^{2}+y^{2}<a^{2}\big\}.
\end{align*}
Suppose that $N\ge1$ particles diffuse in $\Omega$ with diffusivity $D>0$, reflect from $\partial\Omega$, and are initially placed at
\begin{align*}
\x_{0}=(0,0,z_{0})\quad \text{with }z_{0}\in(0,h). 
\end{align*}
We are interested in the first time that any of the $N$ particles hits the target.{ See Figure}~\ref{figcyl}{a for an illustration.}

\begin{figure}
\begin{center}
\includegraphics[width=8cm]{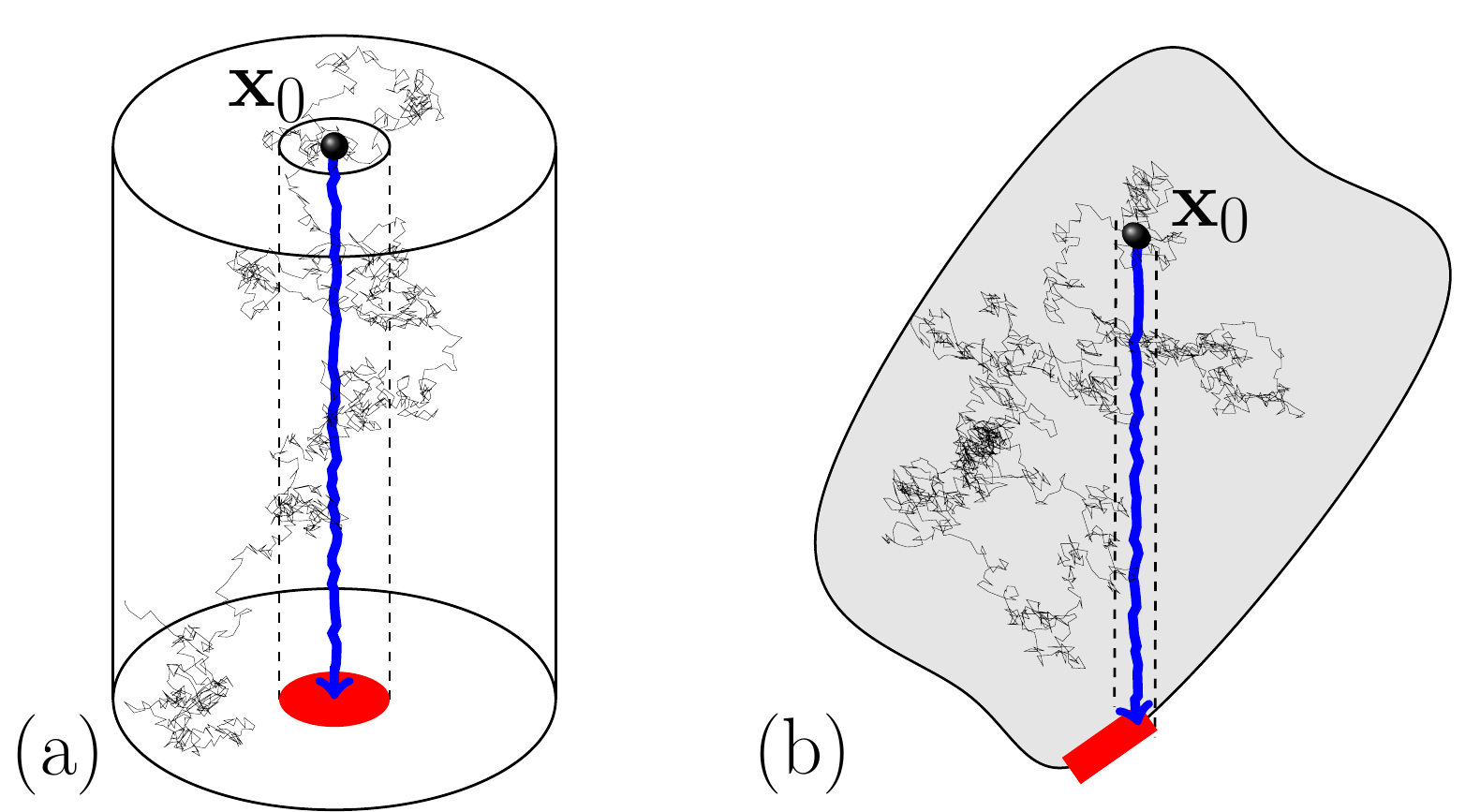}
\caption{
{(a) The simple cylindrical geometry considered in section}~\ref{section cylinder}{. The thin black curve depicts a typical searcher trajectory, while the thick blue curve shows the path of the fastest searcher. (b) The method developed for the cylinder in section}~\ref{section cylinder}{ is extended to more general spatial domains in section}~\ref{section general}{ by proving that the fastest searcher never leaves a thin tube connecting the initial location to the target.}}\label{figcyl}
\end{center}
\end{figure}

Let $Z_{n}(t)\in[0,h]$ and $R_{n}(t)\in[0,r]$ denote the height and radial position of the $n$-th particle at time $t\ge0$. Define the first time that the $n$-th particle reaches the target,
\begin{align*}
\tau_{n}:=\inf\{t>0:Z_{n}(t)=0,\;R_{n}(t)<a\},\quad n\in\{1,\dots,N\}.
\end{align*}
Next, define the first time that the $n$-th particle reaches the bottom of the cylinder (regardless of the radial position),
\begin{align*}
\tau_{n,z}:=\inf\{t>0:Z_{n}(t)=0\},\quad n\in\{1,\dots,N\}.
\end{align*}
Further, define the first time that any particle reaches the target and the first time that any particle reaches the bottom of the cylinder,
\begin{align}
T_{\threed}&:=\min_{n}\{\tau_{n}\}\nonumber\\
T_{\oned}&=T_{\oned}(z_{0}):=\min_{n}\{\tau_{n,z}\}.\label{Tz}
\end{align}
We sometimes write $T_{\oned}(z_{0})$ to emphasize that it is the time for a particle diffusing in 1d to reach a target that is distance $z_{0}>0$ from its initial position.

It is immediate that $T_{\oned}\le T_{\threed}$ with probability one, since $\tau_{n,z}\le\tau_{n}$ for each $n\in\{1,\dots,N\}$. The following theorem shows that the $m$-th moments of $T_{\threed}$ and $T_{\oned}$ for $m\ge1$ become identical as $N$ grows. Throughout this paper, the notation ``$f\sim g$ as $N\to\infty$'' means $\lim_{N\to\infty}f/g=1$.

\begin{theorem}\label{theorem simple}
For any moment $m\ge1$, we have that
\begin{align*}
\E[(T_{\threed})^{m}]
\sim\E[(T_{\oned}(z_{0}))^{m}]
\sim\Big(\frac{z_{0}^{2}}{4D\log N}\Big)^{m}\quad\text{as }N\to\infty.
\end{align*}
\end{theorem}

We make {four} comments on Theorem~\ref{theorem simple}. {First, Theorem}~\ref{theorem simple} {holds for any fixed target size (including a small target)}. Second, the proof of Theorem~\ref{theorem simple} relies on {proving} that the path of the first particle to reach the target is almost a straight line from the initial position to the{ target}. {Third}, while Theorem~\ref{theorem simple} concerns a very specific spatial domain, we extend the argument to much more general spatial domains in Section~\ref{section general} {(see Figure}~\ref{figcyl}{b for an illustration).}

{Fourth, we emphasize that the cylinder is assumed to be finite, meaning $r<\infty$ and $h<\infty$. In fact, if $r=h=\infty$, then the conclusion of the theorem cannot hold since $\E[T_{\threed}]=\infty$ for every $N\ge1$. To see this, note that since Brownian motion is transient in 3d, each particle has a strictly positive probability of never reaching the target,}
\begin{align*}
\P(\tau_{1}=\infty)>0.
\end{align*}
{Therefore, for any $N\ge1$, there is a strictly positive probability that all of the $N$ particles never reach the target,}
\begin{align*}
\P(T_{\threed}=\infty)
=(\P(\tau_{1}=\infty))^{N}>0.
\end{align*}
{Since any random variable that is infinite with strictly positive probability must have infinite expectation, we obtain $\E[T_{\threed}]=\infty$. We note that this is distinct from the phenomenon where the MFPTs of single particles can be infinite while the MFPT of the fastest particle can be finite. Indeed, in the 1d case above with $h=\infty$, it is known that $\E[\tau_{1,z}]=\infty$, but $\E[T_{\oned}]<\infty$ for $N\ge3$} \cite{lindenberg1980}.

\subsection{Some lemmas and the proof of Theorem~\ref{theorem simple}}

Before proving Theorem~\ref{theorem simple}, we outline the basic idea of the proof. Notice that if the first particle that hits the bottom of the cylinder hits the bottom before its radial position escapes the ``inner'' cylinder of radius $a>0$, then $T_{\threed}=T_{\oned}$. Further, since $T_{\oned}\to0$ almost surely as $N\to\infty$, this event happens with high probability since the radial position has little time to escape this inner cylinder before time $T_{\oned}$. Therefore, $\E[(T_{\threed})^{m}]\approx\E[(T_{\oned})^{m}]$ if $N\gg1$.

This first preliminary lemma collects some basic bounds on the short-time behavior of a radial diffusion process (a 2d Bessel process). In particular, it allows us to quantify the probability that a particle quickly escapes the inner cylinder.

\begin{lemma}\label{bound radial}
Let $\tau_{n,\side}$ denote the first time that the $n$-th radial position escapes the disk of radius $a>0$,
\begin{align*}
\tau_{n,\side}
:=\inf\{t>0:R_{n}(t)>a\}.
\end{align*}
That is, $\tau_{n,\side}$ is the first time that the $n$-th particle escapes the inner cylinder of radius $a>0$ through its ``side.''

There exists constants $C_{1},C_{2}>0$ (depending on $a$ and $D$) so that for all $t>0$ sufficiently small,
\begin{align*}
\P(\tau_{1,\side}\le t)
&\le C_{1}t^{-1/2}\exp(-C_{2}t^{-1}).
\end{align*}
\end{lemma}

The next lemma bounds the short-time behavior of a 1d diffusion. In particular, it allows us to quantify the rate that $T_{\oned}$ vanishes as $N\to\infty$.

\begin{lemma}\label{bound z}
There exists constants $C_{3},C_{4}>0$ (depending on $z_{0}$ and $D$) so that for all $t>0$ sufficiently small,
\begin{align*}
\P(\tau_{1,z}>t)
&\le 1-C_{3}t^{1/2}\exp(-C_{4}t^{-1}).
\end{align*}
\end{lemma}

The proofs of Lemmas~\ref{bound radial} and \ref{bound z} are straightforward and collected in the Appendix.

\begin{proof}[Proof of Theorem~\ref{theorem simple}]
Let $\nf\in\{1,\dots,N\}$ denote the index of the fastest particle to reach the bottom of the cylinder. That is,
\begin{align*}
\nf\in\{1,\dots,N:Z_{\nf}(T_{\oned})=0\}.
\end{align*}
Since the event that two or more particles reach the bottom of the cylinder at the same time has probability zero, the index $\nf$ is almost surely unique

Let $A$ be the event that $R_{\nf}(t)$ does not escape the disk of radius $a>0$ before time $T_{\oned}$. That is,
\begin{align*}
A
:=\{\tau_{\nf,\side}>T_{\oned}\}
=\{R_{\nf}(t)<a\text{ for all }t\in[0,T_{\oned}]\}.
\end{align*}
If $A$ happens, then $T_{\threed}=T_{\oned}$. In particular, if $1_{A}$ denotes the indicator function on the event $A$, then raising $T_{\threed}$ and $T_{\oned}$ to the $m$-th power yields
\begin{align*}
T_{\threed}^{m}1_{A}=T_{\oned}^{m}1_{A}\quad \text{almost surely}.
\end{align*}
Letting $A^{c}$ denote the complement of the event $A$, we then have that
\begin{align}\label{start}
\frac{\E[T_{\threed}^{m}]}{\E[T_{\oned}^{m}]}
&=\frac{\E[T_{\oned}^{m}1_{A}]}{\E[T_{\oned}^{m}]}+\frac{\E[T_{\threed}^{m}1_{A^{c}}]}{\E[T_{\oned}^{m}]}
=1+\frac{\E[(T_{\threed}^{m}-T_{\oned}^{m})1_{A^{c}}]}{\E[T_{\oned}^{m}]}.
\end{align}
Since $0<T_{\oned}\le T_{\threed}\le\tau_{1}$ almost surely, the Cauchy Schwarz inequality gives
\begin{align*}
\E[(T_{\threed}^{m}-T_{\oned}^{m})1_{A^{c}}]
\le\E[\tau_{1}^{m}1_{A^{c}}]
\le\sqrt{\P(A^{c})}\sqrt{\E[\tau_{1}^{2m}]},
\end{align*}
{where $\E[\tau_{1}^{2m}]<\infty$, since the domain is bounded.} In view of \eqref{start}, we now show that $\sqrt{\P(A^{c})}/\E[T_{\oned}^{m}]\to0$ as $N\to\infty$ to complete the proof.

Since $Z_{n}$ and $R_{n}$ are independent, we have that
\begin{align}\label{int}
\begin{split}
\P(A^{c})
=\P(\tau_{\nf,\side}\le T_{\oned})
=\P(\tau_{1,\side}\le T_{\oned})
&=\int_{0}^{\infty}\P(\tau_{1,\side}\le t)f(t)\,\dd t,
\end{split}
\end{align}
where $f$ is the probability density of $T_{\oned}$. By independence of the FPTs $\tau_{1,z},\dots,\tau_{N,z}$, we have that
\begin{align*}
f(t):=-S'(t)=N(S_{0}(t))^{N-1}f_{0}(t),
\end{align*}
where $S_{0}(t):=\P(\tau_{1,z}>t)$ is the so-called survival probability of $\tau_{1,z}$, $f_{0}(t):=-S_{0}'(t)$ is the probability density of $\tau_{1,z}$, and $S(t):=\P(T_{\oned}>t)=(S_{0}(t))^{N}$ is the survival probability of $T_{\oned}$.

Returning to \eqref{int}, for each $\delta>0$, we have that
\begin{align*}
\P(A^{c})
=\int_{0}^{\delta}\P(\tau_{1,\side}\le t)f(t)\,\dd t
+\int_{\delta}^{\infty}\P(\tau_{1,\side}\le t)f(t)\,\dd t
=:I_{1}(\delta)+I_{2}(\delta).
\end{align*}
Since $\P(\tau_{1,\side}\le t)$ is an increasing function of time, we have that
\begin{align}\label{i1}
\begin{split}
I_{1}(\delta)
=\int_{0}^{\delta}\P(\tau_{1,\side}\le t)f(t)\,\dd t
\le\P(\tau_{1,\side}\le \delta).
\end{split}
\end{align}
Furthermore, since $S_{0}(t)$ is a decreasing function of time, we have that
\begin{align}\label{i2}
\begin{split}
I_{2}(\delta)
&=\int_{\delta}^{\infty}\P(\tau_{1,\side}\le t)N(S_{0}(t))^{N-1}f_{0}(t)\,\dd t\\
&\le N(S_{0}(\delta))^{N-1}\int_{\delta}^{\infty}\P(\tau_{1,\side}\le t)f_{0}(t)\,\dd t
\le N(S_{0}(\delta))^{N-1}.
\end{split}
\end{align}

Now, setting
\begin{align*}
\delta=\frac{1}{\sqrt{\log N}}
\end{align*}
and using \eqref{i1}-\eqref{i2} and Lemmas~\ref{bound radial} and \ref{bound z}, we have that for $N$ sufficiently large,
\begin{align*}
\P(A^{c})
&\le
C_{1}(\log N)^{1/4}\exp(-C_{2}\sqrt{\log N})\\
&\quad+N\big[1-C_{3}(\log N)^{-1/4}\exp(-C_{4}\sqrt{\log N})\big]^{N-1}.
\end{align*}
Therefore, using that Jensen's inequality \cite{durrett2019} ensures that $(\E[T_{\oned}])^{m}\le\E[T_{\oned}^{m}]$ for $m\ge1$, we then have by \eqref{1d} that
\begin{align*}
\lim_{N\to\infty}\frac{\sqrt{\P(A^{c})}}{\E[T_{\oned}^{m}]}
\le\lim_{N\to\infty}\frac{\sqrt{\P(A^{c})}}{(\E[T_{\oned}])^{m}}=0.
\end{align*}
Therefore, \eqref{start} implies that $\E[T_{\threed}^{m}]\sim\E[T_{\oned}^{m}]$ as $N\to\infty$. Using \eqref{ball} completes the proof.
\end{proof}

\section{General 2d and 3d domains}\label{section general}

In this section, we extend the arguments of section~\ref{section cylinder} to a general class of 2d and 3d domains. To simplify the exposition, we first consider a 3d domain.

Let $\Omega\subset\R^{3}$ be a bounded and open 3d spatial domain. Assume that the boundary, $\partial\Omega$, contains a distinguished region (or a collection of regions), $\partial\Omega_{\text{T}}\subset\partial\Omega$, which we refer to as the target. Let $\partial\Omega_{\text{R}}=\partial\Omega\backslash\partial\Omega_{\text{T}}$ denote the rest of the boundary. 
We assume that $\partial\Omega_{\text{T}}$ and $\partial\Omega_{\text{R}}$ are smooth and that $\partial\Omega_{\text{T}}$ is nonempty and relatively open in $\partial\Omega$.

Consider a set of $N$ independent particles that diffuse in $\Omega\subset\R^{3}$ with diffusivity $D>0$ and reflect from $\partial\Omega$. Suppose the particles are initially placed at $\x_{0}\in\Omega$ and assume there exists a ``target point'' $\x_{\T}\in\partial\Omega_{\text{T}}$ so that the straight line path from $\x_{0}$ to $\x_{\T}$ does not intersect $\partial\Omega$. That is, assume that the following straight line lies entirely in $\Omega$,
\begin{align}\label{sl}
\{(1-s)\x_{0}+s\x_{\text{T}}\in\R^{3}:s\in[0,1)\}
\subset\Omega.
\end{align}
We note that \eqref{sl} ensures that there are no ``obstacles'' in the straight line path between $\x_{0}$ and $\x_{\T}$. 

Without loss of generality, we choose a Cartesian coordinate system $(x,y,z)$ so that (i) the target point is the origin,
\begin{align*}
\x_{\T}=(0,0,0),
\end{align*}
and (ii) the initial particle location is directly ``above'' the origin in the $z$ direction,
\begin{align*}
\x_{0}=(0,0,z_{0})\quad\text{with }z_{0}:=\|\x_{0}-\x_{\text{T}}\|>0.
\end{align*}
{ See Figure}~\ref{figcyl}{b for an illustration.}

Let
\begin{align*}
\X_{n}(t)=(X_{n}(t),Y_{n}(t),Z_{n}(t))
\in\overline{\Omega}\subset\R^{3},\quad n\in\{1,\dots,N\},
\end{align*}
denote the position of the $n$-th particle at time $t\ge0$. Define the first time that the $n$-th particle reaches the target,
\begin{align*}
\tau_{n}
:=\inf\{t>0:\X_{n}(t)\in\partial\Omega_{\text{T}}\},\quad n\in\{1,\dots,N\},
\end{align*}
and the first time any particle reaches the target,
\begin{align}\label{T}
T_{\threed}:=\min_{n}\{\tau_{n}\}.
\end{align}

For large $N$, we prove that the $m$-th moment of the FPT \eqref{T} of the fastest particle in this general 3d geometry is bounded above by the $m$-th moment of the FPT in a corresponding 1d system, for any moment $m\ge1$. That is, the following theorem ensures that $\E[(T_{\threed})^{m}]$ cannot be greater than $\E[(T_{\oned}(z_{0}))^{m}]$ for large $N$, where $z_{0}$ is the distance between $\x_{0}$ and any point $\x_{\T}\in\partial\Omega_{\T}$, assuming the line from $\x_{0}$ to $\x_{\T}$ lies in $\Omega$ (see \eqref{sl}).

\begin{theorem}\label{theorem upper}
Let $T_{\threed}$ be the FPT in \eqref{T} for the general 3d geometry described above. Let $T_{\oned}(z_{0})$ be the FPT defined in \eqref{Tz} for a 1d system, where $z_{0}=\|\x_{0}-\x_{\T}\|$ is the distance between the initial particle location $\x_{0}$ and any point $\x_{\T}$ in the target $\partial\Omega_{\T}$, assuming \eqref{sl} holds. Then, for any moment $m\ge1$, we have that
\begin{align*}
\limsup_{N\to\infty}\frac{\E[(T_{\threed})^{m}]}{\E[(T_{\oned}(z_{0}))^{m}]}
=\limsup_{N\to\infty}\left(\frac{z_{0}^{2}}{4D\log N}\right)^{-m}\E[(T_{\threed})^{m}]
\le 1.
\end{align*}
\end{theorem}

To find the exact asymptotic behavior of $\E[(T_{\threed})^{m}]$ (rather than the upper bound in Theorem~\ref{theorem upper}), let $\x_{\T}^{*}\in\overline{\partial\Omega_{\T}}$ minimize the distance to $\x_{0}$. That is, assume $\x_{\T}^{*}\in\overline{\partial\Omega_{\T}}$ is such that
\begin{align}\label{min}
\|\x_{0}-\x_{\T}^{*}\|
\le\|\x_{0}-\x_{\T}\|\quad\text{for all }\x_{\T}\in\partial\Omega_{\T}.
\end{align}
Further, assume \eqref{sl} holds with $\x_{\T}$ replaced by $\x_{\T}^{*}$,
\begin{align}\label{slstar}
\{(1-s)\x_{0}+s\x_{\T}^{*}\in\R^{3}:s\in[0,1)\}
\subset\Omega.
\end{align}
In addition, letting
\begin{align}\label{zd}
z_{0}
:=\|\x_{0}-\x_{\T}^{*}\|,
\end{align}
we also assume that the region
\begin{align}\label{S}
S
:=\{\x\in\R^{3}:\|\x-\x_{0}\|< z_{0}\}\cap\Omega\in\R^{3}
\end{align}
is a so-called star domain, meaning that the line from $\x_{0}$ to $\x$ lies in $S$ for all $\x\in S$,
\begin{align}\label{star}
\{(1-s)\x_{0}+s\x\in\R^{3}:s\in[0,1]\}
\subset S\quad\text{for all }\x\in S.
\end{align}
{We note that assuming the domain $\Omega$  is convex ensures that }\eqref{star}{ holds (though }\eqref{star}{ can hold even if $\Omega$ is not convex).}

\begin{theorem}\label{theorem all}
If \eqref{min}-\eqref{star} are satisfied, then for any moment $m\ge1$ we have that
\begin{align*}
\E[(T_{\threed})^{m}]
\sim\E[(T_{\oned}(z_{0}))^{m}]
\sim\Big(\frac{z_{0}^{2}}{4D\log N}\Big)^{m}
\quad\text{as }N\to\infty.
\end{align*}
\end{theorem}

\subsection{Some lemmas and the proofs of Theorems~\ref{theorem upper} and \ref{theorem all}}

The proof of Theorem~\ref{theorem upper} generalizes the methods developed in Section~\ref{section cylinder}. The essential idea is to show that the fastest particle does not leave a thin cylinder encapsulating the straight line from $\x_{0}$ to $\x_{\T}$ in \eqref{sl}.

Let $\Sigma_{r,b,c}$ denote the following cylinder centered on the $z$-axis with radius $r>0$, bottom at $z=b$, and top at $z=c>b$,
\begin{align}\label{cyl}
\Sigma_{r,b,c}
:=\{(x,y,z)\in\R^{3}:x^{2}+y^{2}< r^{2},\,z\in(b,c)\}.
\end{align}
By taking the radius $r>0$ sufficiently small, we choose $c$ and $b$ so that
\begin{align*}
b\le0<z_{0}<c,\quad
\Sigma_{r,z_{0},c}\subset\Omega,\quad\text{and}\quad
\Sigma_{r,b,c}\cap\partial\Omega\subset\partial\Omega_{\T}.
\end{align*}
Furthermore, we choose $b$ so that the bottom of the cylinder is completely ``below'' the target, meaning
\begin{align*}
b\le\inf\{z:(x,y,z)\in\Sigma_{r,b,z_{0}}\cap\partial\Omega_{\T}\}.
\end{align*}

Define the first time that the $n$-th particle hits either the sides or the top of the cylinder $\Sigma_{r,b,c}$,
\begin{align*}
\tau_{n,\side}
&:=\inf\{t>0:X_{n}^{2}(t)+Y_{n}^{2}(t)=r^{2}\},\\
\tau_{n,\top}
&:=\inf\{t>0:Z_{n}(t)=c\}.
\end{align*}
To bound $T_{\threed}$, we define a new process $\Z_{n}(t)$ which is equal to $Z_{n}(t)$ before the $n$-th particle hits either the target or escapes the cylinder $\Sigma_{r,b,c}$, and then diffuses independently. Specifically, define the stopping time
\begin{align*}
\tau_{n,\esc}
:=\min\{\tau_{n},\tau_{n,\side},\tau_{n,\top}\},
\end{align*}
and
\begin{align*}
\Z_{n}(t)
=\begin{cases}
Z_{n}(t) & t\le\tau_{n,\esc},\\
Z_{n}(\tau_{n,\esc})+\sqrt{2D}W_{n}(t) & t>\tau_{n,\esc},
\end{cases}
\end{align*}
where $W_{n}$ denotes an independent standard Brownian motion. Next, define the first time that $\Z_{n}$ hits the top or bottom of the cylinder $\Sigma_{r,b,c}$,
\begin{align*}
\wtau_{n,\top}
&:=\inf\{t>0:\Z_{n}(t)=c\},\\
\wtau_{n,\bot}
&:=\inf\{t>0:\Z_{n}(t)=b\}.
\end{align*}
Next, define
\begin{align*}
\tau_{n}^{+}
:=\begin{cases}
\wtau_{n,\bot} & \wtau_{n,\bot}<\min\{\tau_{n,\side},\wtau_{n,\top}\},\\
\min\{\tau_{n,\side},\wtau_{n,\top}\}+1 & \text{otherwise}.
\end{cases}
\end{align*}
Further, define
\begin{align*}
T_{+}
:=\min_{n}\{\tau_{n}^{+}\}.
\end{align*}

The next lemma shows that the $m$-th moment of $T$ is bounded above by the $m$-th moment of $T_{+}$ for sufficiently large $N$.
\begin{lemma}\label{lemma reflect}
If $m\ge1$, then
\begin{align*}
\limsup_{N\to\infty}\frac{\E[(T_{\threed})^{m}]}{\E[(T_{+})^{m}]}\le1.
\end{align*}
\end{lemma}

Next, we use the methods of Section~\ref{section cylinder} to show that the large $N$ behavior of $\E[T_{+}]$ is unchanged if we make the particles reflect at the sides and top of the cylinder. The proof is similar to the argument of Section~\ref{section cylinder}.
\begin{lemma}\label{lemma same}
If $m\ge1$, then
\begin{align*}
\E[(T_{+})^{m}]
\sim\E[(T_{\oned}(z_{0}+|b|))^{m}]
\sim\Big(\frac{(z_{0}+|b|)^{2}}{4D\log N}\Big)^{m}\quad\text{as }N\to\infty.
\end{align*}
\end{lemma}

\begin{proof}[Proof of Theorem~\ref{theorem upper}]
Combining Lemmas~\ref{lemma reflect} and \ref{lemma same}, we have that
\begin{align*}
\limsup_{N\to\infty}\frac{\E[(T_{\threed})^{m}]}{\E[(T_{\oned}(z_{0}+|b|))^{m}]}
\le 1.
\end{align*}
Since we can {take} the bottom $b$ of the cylinder in \eqref{cyl} arbitrarily close to zero by taking the radius $r$ of the cylinder small, and by the asymptotic behavior in \eqref{ball}, the proof is complete.
\end{proof}

The next lemma bounds $T_{\threed}$ below by the escape time from a sphere of radius $z_{0}$.

\begin{lemma}\label{lemma ball}
Assume \eqref{min}-\eqref{star} are satisfied. For each $n\in\{1,\dots,N\}$, let $R_{n}^{0}(t)$ denote a 3d Bessel process satisfying
\begin{align*}
\dd R_{n}^{0}(t)
=\frac{2D}{R_{n}^{0}(t)}\,\dd t+\sqrt{2D}\,\dd W_{n}(t),\quad R_{n}^{0}(0)=0,
\end{align*}
where $\{W_{1},\dots,W_{N}\}$ are independent, standard one-dimensional Brownian motions. Define first time that $R_{n}^{0}$ hits radius $z_{0}>0$,
\begin{align*}
\tau_{n,\ball}^{0}
:=\inf\{t>0:R_{n}^{0}(t)\ge z_{0}\},
\end{align*}
and the first time any $R_{n}^{0}$ hits radius $z_{0}>0$,
\begin{align*}
T_{\ball}^{0}
:=\min_{n}\{\tau_{n,\ball}^{0}\}.
\end{align*}
Then
\begin{align*}
\P(T_{\ball}^{0}\ge t)\le \P(T_{\threed}\ge t)\quad\text{for each }t\ge0.
\end{align*}
\end{lemma}

With these lemmas in place, the proof of Theorem~\ref{theorem all} follows quickly.

\begin{proof}[Proof of Theorem~\ref{theorem all}]
By Lemma~\ref{lemma ball}, we have that
\begin{align*}
\Big(\frac{z_{0}^{2}}{4D\log N}\Big)^{-m}\E[(T_{\ball}^{0})^{m}]
\le\Big(\frac{z_{0}^{2}}{4D\log N}\Big)^{-m}\E[(T_{\threed})^{m}]\quad\text{for each }N\ge1.
\end{align*}
Therefore, \eqref{ball} implies that
\begin{align*}
1
=\liminf_{N\to\infty}\Big(\frac{z_{0}^{2}}{4D\log N}\Big)^{-m}\E[(T_{\ball}^{0})^{m}]
&\le\liminf_{N\to\infty}\Big(\frac{z_{0}^{2}}{4D\log N}\Big)^{-m}\E[(T_{\threed})^{m}]\\
&\le\limsup_{N\to\infty}\Big(\frac{z_{0}^{2}}{4D\log N}\Big)^{-m}\E[(T_{\threed})^{m}]. 
\end{align*}
Applying Theorem~\ref{theorem upper} completes the proof.
\end{proof}

\subsection{Space dimension $d=2$}\label{section 2d}

The arguments above can be immediately applied to 2d spatial domains. We now briefly state the results. 

Let $\Omega\subset\R^{2}$ be a bounded and open 2d spatial domain with a target $\partial\Omega_{\T}\subseteq\partial\Omega$. Assume that $\partial\Omega_{\text{T}}$ and $\partial\Omega_{\text{R}}:=\partial\Omega\backslash\partial\Omega_{\T}$ are smooth and that $\partial\Omega_{\text{T}}$ is nonempty and relatively open in $\partial\Omega$. Suppose $N$ independent particles diffuse in $\Omega\subset\R^{2}$ with diffusivity $D>0$ and reflect from $\partial\Omega$. Suppose the particles are initially placed at $\x_{0}\in\Omega$ and assume there exists a ``target point'' $\x_{\T}\in\partial\Omega_{\text{T}}$ so that
\begin{align}\label{sl2}
\{(1-s)\x_{0}+s\x_{\text{T}}\in\R^{2}:s\in[0,1)\}
\subset\Omega.
\end{align}

Let $\X_{n}(t)\in\overline{\Omega}\subset\R^{2}$ denote the position of the $n$-th particle at time $t\ge0$. Define the first time that the $n$-th particle reaches the target,
\begin{align*}
\tau_{n}
:=\inf\{t>0:\X_{n}(t)\in\partial\Omega_{\text{T}}\},\quad n\in\{1,\dots,N\},
\end{align*}
and the first time any particle reaches the target,
\begin{align*}
T_{\twod}:=\min_{n}\{\tau_{n}\}.
\end{align*}
We now state the 2d analog of Theorem~\ref{theorem upper}.
\begin{theorem}\label{theorem upper2}
Let $T_{\oned}(z_{0})$ be the FPT defined in \eqref{Tz} for a 1d system, where $z_{0}=\|\x_{0}-\x_{\T}\|$ is the distance between the initial particle location $\x_{0}$ and any point $\x_{\T}$ in the target $\partial\Omega_{\T}$, assuming \eqref{sl2} holds. Then, for any moment $m\ge1$, we have that
\begin{align*}
\limsup_{N\to\infty}\frac{\E[(T_{\twod})^{m}]}{\E[(T_{\oned}(z_{0}))^{m}]}
=\limsup_{N\to\infty}\left(\frac{z_{0}^{2}}{4D\log N}\right)^{-m}\E[(T_{\twod})^{m}]
\le 1.
\end{align*}
\end{theorem}

To state the 2d analog of Theorem~\ref{theorem all}, let $\x_{\T}^{*}\in\overline{\partial\Omega_{\T}}$ minimize the distance to $\x_{0}$, and assume \eqref{sl2} holds with $\x_{\T}$ replaced by $\x_{\T}^{*}$. In addition, let $z_{0}:=\|\x_{0}-\x_{\T}^{*}\|$, and assume that the region
\begin{align*}
S
:=\{\x\in\R^{2}:\|\x-\x_{0}\|< z_{0}\}\cap\Omega\in\R^{2}
\end{align*}
is a star domain, meaning
\begin{align}\label{star2}
\{(1-s)\x_{0}+s\x\in\R^{2}:s\in[0,1]\}
\subset S\quad\text{for all }\x\in S.
\end{align}
{Again, we note that assuming the domain $\Omega$ is convex ensures that }\eqref{star2}{ holds (though }\eqref{star2}{ can hold even if $\Omega$ is not convex).}

\begin{theorem}\label{theorem all2}
Under the assumptions just given, for any moment $m\ge1$ we have that
\begin{align*}
\E[(T_{\twod})^{m}]
\sim\E[(T_{\oned}(z_{0}))^{m}]
\sim\Big(\frac{z_{0}^{2}}{4D\log N}\Big)^{m}
\quad\text{as }N\to\infty.
\end{align*}
\end{theorem}

\section{Discussion} 

{In this paper,} we determined the asymptotic behavior of the fastest FPT for a collection of $N\gg1$ identically distributed Brownian searchers to reach a target. Our results hold in a general class of 2d and 3d spatial domains that may have multiple targets that can be on the outer boundaries or inner boundaries of the domain (see Figure~\ref{figextreme}b). Our main conclusion is that the leading order behavior of the $m$-th moment of this fastest FPT does not depend on the details of the spatial domain (including the dimension), but rather only depends on the distance between the initial searcher location and the nearest target.

We proved our results assuming (i) that there exists a straight line path from the initial particle location to a target (see \eqref{sl}) and (ii) that a certain so-called star condition holds (see \eqref{star}). These assumptions allowed us to give relatively short proofs of our results. {However, we recently developed an alternative approach which proves that these assumptions are superfluous }\cite{lawley2019uni}{. Indeed, this more recent work proves that the asymptotic behavior in }\eqref{1d}{ holds under very general conditions, including (i) diffusions in $\R^{d}$ with space-dependent diffusivities and drift fields and (ii) diffusions on $d$-dimensional smooth Riemannian manifolds that may contain reflecting obstacles} \cite{lawley2019uni}. {In addition, we have also recently determined how the details of the spatial domain can affect the fastest FPT statistics at second order} \cite{lawley2019gumbel}. {Further, the results in }\cite{lawley2019gumbel}{ describe the limiting probability distribution of the fastest FPTs in terms of a certain Gumbel distribution involving the so-called LambertW function. For recent works emphasizing the importance of the full first passage time distribution, see }\cite{grebenkov2018, grebenkov2018b, lawley2019imp, bernoff2018boundary, lawley2019patchy}.

More generally, since the fastest FPT is the minimum of a large set of iid random variables, investigating its distribution falls into the field of \emph{extreme statistics} \cite{gumbel2012}. The theory of extreme statistics has been heavily used in disciplines such as finance, engineering, and earth sciences \cite{coles2001, novak2011}, but the theory is relatively unknown in biology. We expect that extreme statistics will soon find many applications in biology. Indeed, several recent commentaries have made this prediction and suggested various applications \cite{basnayake2019, coombs2019, redner2019, sokolov2019, rusakov2019, martyushev2019, tamm2019}.

Finally, a distinguishing aspect of our approach to extreme FPT theory in this paper is the use of probabilistic methods. In contrast, prior work has tended to employ asymptotic and perturbation methods \cite{basnayake2019, meerson2015, yuste_diffusion_2000, yuste2001, ro2017}. Our approach thus follows some other recent studies that have applied probabilistic methods to Brownian escape problems that are more commonly studied by PDE asymptotics \cite{PB3, PB4, PB2} (see also \cite{lawley2018prob, handy2018, handy2019}). Going forward, we anticipate that techniques from probability theory will continue to be useful in extreme FPT theory.

\section{Appendix}

In this Appendix, we collect the proofs of all the lemmas.

\begin{proof}[Proof of Lemma~\ref{bound radial}]
First notice that $R_{n}(t)$ is a 2d Bessel process. In particular, $R_{n}(t)$ has the same law as $\sqrt{2D}\sqrt{W_{1}^{2}(t)+W_{2}^{2}(t)}$, where $W_{1}(t)\in\R$ and $W_{2}(t)\in\R$ are independent standard Brownian motions. Next, define the 3d Bessel process,
\begin{align*}
R^{\text{3d}}(t)
:=\sqrt{2D}\sqrt{W_{1}^{2}(t)+W_{2}^{2}(t)+W_{3}^{2}(t)},
\end{align*}
where $W_{3}(t)\in\R$ is a third independent standard Brownian motion. If $\tau$ is the first time that $R^{\text{3d}}(t)$ hits radius $a>0$,
\begin{align*}
\tau
:=\inf\{t>0:R^{\text{3d}}(t)=a\},
\end{align*}
then it is immediate that $\P(\tau_{1,\side}\le t)\le\P(\tau\le t)$ for all $t\ge0$. Since $\P(\tau\le t)$ has the explicit formula (see, for example, \cite{bernoff2018}),
\begin{align*}
\P(\tau\le t)
=\sqrt{\frac{4a^{2}}{\pi Dt}}\sum_{k=0}^{\infty}e^{-a^{2}(n+\frac{1}{2})^{2}/(Dt)},
\end{align*}
we may take $C_{1}=4a(\pi D)^{-1/2}$ and $C_{2}=a^{2}/(4D)$ to complete the proof.
\end{proof}

\begin{proof}[Proof of Lemma~\ref{bound z}]
Recall that $\tau_{1,z}$ is the first time $Z_{1}(t)$ hits $z=0$, assuming $Z_{1}(0)=z_{0}>0$ and $Z_{1}(t)$ reflects from $z=h\ge z_{0}$. It is immediate that
\begin{align*}
\P(\tau_{1,z}>t)
\le\P(\tau>t),
\end{align*}
where $\tau$ is the first time a standard Brownian motion, $W(t)\in\R$, hits $z_{0}/\sqrt{2D}$,
\begin{align*}
\tau:=\inf\{t>0:W(t)=z_{0}/\sqrt{2D}\}.
\end{align*}
The reflection principle \cite{durrett2019} then gives
\begin{align*}
\P(\tau>t)
=1-\P(\tau\le t)
=1-2\P(W(t)\ge z_{0}/\sqrt{2D})
=\text{erf}(z_{0}/\sqrt{4Dt}),
\end{align*}
where $\text{erf}(x)$ denotes the error function, which has the large $x$ behavior,
\begin{align}\label{erf}
\text{erf}(x):=\frac{1}{\sqrt{\pi}}\int_{-x}^{x}e^{-s^{2}}\,\dd s
=1-\frac{e^{-x^{2}}}{x\sqrt{\pi}}\Big(1-\frac{2}{x^{2}}+\O(x^{-4})\Big)\quad\text{as }x\to\infty.
\end{align}
Hence, we may take $C_{3}=\frac{\sqrt{D}}{\sqrt{\pi}z_{0}}$ and $C_{4}=z_{0}^{2}/(4D)$ to complete the proof.
\end{proof}

\begin{proof}[Proof of Lemma~\ref{lemma reflect}]
Define the event
\begin{align*}
B
&:=\{\tau_{1}<\min\{\tau_{1,\side},\tau_{1,\top}\}\},
\end{align*}
and let $B^{c}$ denote its complement. It is immediate that
\begin{align*}
\P(\tau_{1}>t|B)
\le\P(\tau_{1}^{+}>t|B),
\end{align*}
since $\tau_{1}\le\tau_{1}^{+}$ if $\tau_{1}<\min\{\tau_{1,\side},\tau_{1,\top}\}$. Next, by the definition of $\tau_{1}^{+}$, we have that
\begin{align*}
\P(\tau_{1}^{+}>t|B^{c})=1\quad\text{for all }t<1.
\end{align*}
Therefore
\begin{align*}
\P(\tau_{1}>t)
\le\P(\tau_{1}^{+}>t|B)\P(B)
+\P(B^{c})=\P(\tau_{1}^{+}>t)\quad\text{for all }t<1.
\end{align*}

Now, it is an identity that
\begin{align*}
\E[T_{\threed}^{m}]
=
\int_{0}^{1}(\P(\tau_{1}>t^{1/m}))^{N}\,\dd t
+\int_{1}^{\infty}(\P(\tau_{1}>t^{1/m}))^{N}\,\dd t.
\end{align*}
Therefore, we have that
\begin{align*}
\E[T_{\threed}^{m}]
&=
\int_{0}^{1}(\P(\tau_{1}>t^{1/m}))^{N}\,\dd t
+\int_{1}^{\infty}(\P(\tau_{1}>t^{1/m}))^{N}\,\dd t\\
&\le
\int_{0}^{1}(\P(\tau_{1}^{+}>t^{1/m}))^{N}\,\dd t
+(\P(\tau_{1}>1))^{N}\int_{1}^{\infty}\frac{\P(\tau_{1}>t^{1/m})}{\P(\tau_{1}>1)}\,\dd t.
\end{align*}
Therefore, Jensen's inequality yields
\begin{align*}
\frac{\E[T_{\threed}^{m}]}{\E[T_{+}^{m}]}
&\le
\frac{\int_{0}^{1}(\P(\tau_{1}^{+}>t^{1/m}))^{N}\,\dd t
+(\P(\tau_{1}>1))^{N}\int_{1}^{\infty}\frac{\P(\tau_{1}>t^{1/m})}{\P(\tau_{1}>1)}\,\dd t}{(\E[T_{+}])^{1/m}}\\
&\le1+\frac{(\P(\tau_{1}>1))^{N}}{(\E[T_{+}])^{1/m}}\int_{1}^{\infty}\frac{\P(\tau_{1}>t^{1/m})}{\P(\tau_{1}>1)}\,\dd t.
\end{align*}
To complete the proof, we need only that $\E[T_{+}]$ decays slower than $(\P(\tau_{1}>1))^{N}$ as $N\to\infty$, which is implied by Lemma~\ref{lemma same}.
\end{proof}

\begin{proof}[Proof of Lemma~\ref{lemma same}]

Let $T_{\oned}$ be the first time that any $\Z_{n}$ hits $z=b$,
\begin{align*}
T_{\oned}
:=\min_{n}\{\wtau_{n,\bot}\},
\end{align*}
and let $\nf\in\{1,\dots,N\}$ denote the random index of this fastest $\Z_{n}$ particle
\begin{align*}
\nf\in\{1,\dots,N:\Z_{\nf}(T_{\oned})=b\}.
\end{align*}
Define the event that the fastest particle first escapes the cylinder through the bottom
\begin{align*}
A
:=\{\tau_{\nf,\side}>T_{\oned}\}\cap\{\wtau_{\nf,\top}>T_{\oned}\}
=\{T_{+}=T_{\oned}\}.
\end{align*}
Since
\begin{align*}
T_{+}1_{A}=T_{\oned}1_{A}\quad\text{almost surely},
\end{align*}
we can apply the same argument as in Section~\ref{section cylinder} if we can manage to show that $\sqrt{\P(A^{c})}(\log N)^{m}\to0$ as $N\to\infty$.

By De Morgan's laws, we have that
\begin{align*}
\P(A^{c})\le\P(\tau_{\nf,\side}\le T_{\oned})+\P(\wtau_{\nf,\top}\le T_{\oned}).
\end{align*}
The first probability can be handled by the same argument as in Section~\ref{section cylinder}. 

To handle $\P(\wtau_{\nf,\top}\le T_{\oned})$, note that
\begin{align*}
\P(\wtau_{\nf,\top}\le T_{\oned})
\le\P(\wtau_{\nf,\top}\le\delta)+\P(T_{\oned}\ge\delta)\quad\text{for any }\delta>0.
\end{align*}
As in the proof of Theorem~\ref{theorem simple}, we may bound $\P(T_{\oned}\ge\delta)$ for sufficiently small $\delta>0$ by
\begin{align*}
\P(T_{\oned}\ge\delta)
=\int_{\delta}^{\infty}N(S_{0}(t))^{N-1}f_{0}(t)\,\dd t
&\le N(S_{0}(\delta))^{N-1}\\
&\le N\big[1-C_{3}\delta^{1/2}\exp(-C_{4}/\delta)\big]^{N-1}.
\end{align*}

Furthermore, it is immediate that
\begin{align*}
\P(\wtau_{\nf,\top}\le\delta)
\le\P(\wtau_{1,\top}\le\delta).
\end{align*}
Now, notice that $\wtau_{1,\top}$ has the same law as the first time a standard Brownian motion, $W(t)\in\R$, hits $(h-z_{0})/\sqrt{2D}$. In particular, the reflection principle \cite{durrett2019} gives
\begin{align*}
\P(\wtau_{1,\top}\le t)
=2\P(W(t)\ge (h-z_{0})/\sqrt{2D})
=1-\text{erf}((h-z_{0})/\sqrt{4Dt}),
\end{align*}
where $\text{erf}(x)$ denotes the error function with large $x$ behavior given in \eqref{erf}. Hence,
\begin{align*}
\P(\wtau_{\nf,\top}\le\delta)
\le\frac{2}{\sqrt{\pi}}\frac{\sqrt{4D\delta}}{(h-z_{0})}e^{-(h-z_{0})^{2}/(4D\delta)}\quad\text{for sufficiently small }\delta>0.
\end{align*}
Taking $\delta=(\log N)^{-1/2}$ completes the proof.
\end{proof}

\begin{proof}[Proof of Lemma~\ref{lemma ball}]
Define the first time that the $n$-th particle leaves a ball centered at $\x_{0}$ of radius $z_{0}>0$,
\begin{align*}
\tau_{n,\ball}
:=\inf\{t>0:\|\X_{n}(t)-\x_{0}\|\ge z_{0}\}.
\end{align*}
It is immediate that the $n$-th particle cannot reach the target $\partial\Omega_{\T}$ before time $\tau_{n,\ball}$. That is,
\begin{align}\label{balltime}
\tau_{n,\ball}\le\tau_{n}\quad\text{almost surely for each }n\in\{1,\dots,N\}.
\end{align}

Define the radial process
\begin{align*}
R_{n}(t)
:=\|\X_{n}(t)-\x_{0}\|\quad\text{for }t\ge0.
\end{align*}
Notice that $R_{n}$ is not a three-dimensional Bessel process, due to the reflecting boundary $\partial\Omega$. However, $R_{n}$ does satisfy the following stochastic differential equation,
\begin{align}\label{rsde}
\dd R_{n}(t)
=\dd R_{n}^{0}(t)+\n_{R}(\X_{n}(t))\,\dd L_{n}(t),
\end{align}
where $R_{n}^{0}$ is a 3d Bessel process. In \eqref{rsde}, $\n_{R}$ is the radial component of the inner normal field $\n:\partial\Omega\mapsto\R^{3}$, and $L_{n}(t)$ is the local time of $\X_{n}(t)$ on $\partial\Omega$. More precisely, $L_{n}(t)$ is nondecreasing and increases only when $\X_{n}(t)$ is on $\partial\Omega$. The significance of the local time term in \eqref{rsde} is that it forces $\X_{n}(t)$ to reflect from $\partial\Omega$.

By our assumption in \eqref{star} that $S$ in \eqref{S} is a star domain, we are assured that $\n_{R}(\X_{n}(t))\le0$ for all $t\le\tau_{n,\ball}$. Hence, $R_{n}(t)\le R_{n}^{0}(t)$ for all $t\le\tau_{n,\ball}$. Therefore, 
$\tau_{n,\ball}^{0}\le\tau_{n,\ball}$ almost surely. Hence, $T_{\ball}^{0}\le T_{\threed}$ almost surely by \eqref{balltime}.
\end{proof}

\subsubsection*{Acknowledgments}
SDL was supported by the National Science Foundation (Grant Nos.\ DMS-1814832 and DMS-1148230). JBM was supported by the National Science Foundation (Grant No.\ DMS-1814832).

\bibliography{library}
\bibliographystyle{unsrt}
\end{document}